%% file: bc-2core.tex
\documentclass{amsart}

\title{A Note on Computing Betweenness Centrality \\ from the 2-core}
\author{Charalampos E. Tsourakakis}
\address{RelationalAI, Inc.}
\email{charalampos.tsourakakis@relationalai.com}

\usepackage{amsthm}
\usepackage{amssymb}
\usepackage{url}
\usepackage[...]{youngtab}
\usepackage{graphicx}
\usepackage{enumitem}
\usepackage{empheq}
\usepackage{xcolor}
\usepackage{hyperref}

\usepackage{algorithm}
\usepackage{algpseudocode}
 \usepackage{amsmath}
 
\usepackage{geometry}

\newcommand{\beql}[1]{\begin{equation}\label{#1}}

\newcommand{\beq}[1]{\begin{equation}\label{#1}}
\newcommand{\eeq}{\end{equation}}

\newcommand{\Prob}[1]{\ensuremath{{\bf{Pr}}\left[{#1}\right]}}
\newcommand{\Mean}[1]{\ensuremath{{\mathbb E}\left[{#1}\right]}}

\newcommand{\Var}[1]{{\mathbb Var}\left[{#1}\right]}

\newtheorem{thm}{Theorem}[section]
\newtheorem{cor}[thm]{Corollary}

\newtheorem{proposition}[thm]{Proposition}
\newtheorem{definition}[thm]{Definition}

\newcommand{\mybc}{\mathrm{bc}}

\newcommand{\bc}[2][]{\mathrm{bc}_{#1}(#2)}

\newcommand{\degree}[2]{\mathrm{deg}_1^{(#2)}(#1)}

\begin{document}

 \maketitle

\input{src/abstract}

\section{Introduction}
\label{sec:intro}

\input{src/intro}

\section{Related Work}
\label{sec:related}
\input{src/related}

\section{Exact Betweenness Centrality Computation}
\label{sec:exact}

\input{src/proposed}

\section{Sampling after Peeling}
\label{sec:sampling}

\input{src/sampling}

\section{Experiments}
\label{sec:exp}
\input{src/exp}

\section{Conclusion}
\label{sec:concl}
\input{src/concl}

\bibliographystyle{abbrv}
\bibliography{ref}
 

\section{Appendix}
\label{sec:appendix}
\input{src/appendix}

\end{document}

%% file: src/abstract.tex
\begin{abstract}
A central task in network analysis is to identify important nodes in a graph. Betweenness centrality (BC) is a popular centrality measure  that captures the significance of nodes based on the number of shortest paths each node intersects with~\cite{brandes2008variants,freeman1977set}. In this note, we derive a recursive formula to compute the betweenness centralities of a graph from the betweenness centralities of its 2-core.
Furthermore, we analyze mathematically the significant impact of removing degree-one nodes on the estimation of betweenness centrality within the context of the popular pivot sampling scheme for Single-Source Shortest Path (SSSP) computations, as described in the Brandes-Pich approach~\cite{brandes2007centrality,jacob2005algorithms} and implemented in widely used software such as NetworkX. We demonstrate both theoretically and empirically that removing degree-1 nodes can reduce the sample complexity needed to achieve better accuracy, thereby decreasing the overall runtime.
\end{abstract}

%% file: src/intro.tex
Betweenness Centrality (BC)~\cite{anthonisse1971rush,freeman1977set} is a widely used measure in the analysis of complex networks, including applications in social networks and web page search algorithms~\cite{koschutzki2005centrality}.  It is routinely calculated in graphs to assess the significance of nodes and forms a component of various commercial knowledge graph databases such as RelationalAI~\cite{rel} and Neo4j~\cite{neo4j} and it has been a central topic in the field of graph analytics~\cite{riondato2014fast,riondato2018abra,brandes2001faster,brandes2007centrality,jacob2005algorithms}.  Its computation is based on shortest paths. Specifically, for each (ordered) pair of vertices in a connected graph, there are possibly many shortest paths between them. Betweenness centrality of a vertex \( v \) quantifies the number of times a node acts as an intermediary along the shortest path between two other nodes. Formally, given a graph $G(V,E)$ the betweenness centrality is defined as follows:

\begin{definition}[{\bf Betweenness Centrality~\cite{freeman1977set} }]
\label{def:bc}
The Betweenness Centrality (BC) score $\bc{v}$ of a vertex $v$ in the unweighted, possibly directed graph $G=(V,E)$ is:
\begin{equation}
\mathrm{bc}(v) = \frac{1}{(n-1)(n-2)}  \sum_{s,t\in V, s\not=t\not=v} \delta_{st}(v),
\end{equation}
where $\delta_{st}(v)$ is a pair-dependency, namely the fraction of $s-t$ shortest paths passing through node $v$:
\begin{equation}
\delta_{st}(v) = \frac{\sigma_{st}(v)}{\sigma_{st}}.
\end{equation}

\noindent Here, $\sigma_{st}(v)$ denotes the number of the shortest paths from $s$ to $t$ that pass through $v$, and $\sigma_{st}$ denotes the total number of the shortest paths from $s$ to $t$.
\end{definition}

\noindent Computing the betweenness centrality scores can be computationally demanding. The most efficient algorithm for determining exact BC values to date was developed by Brandes, which operates in \(O(nm)\) time for unweighted graphs and \(O(nm + n^2 \log n)\) for weighted graphs~\cite{brandes2001faster}. Here, $n, m$ represent the number of nodes and edges in the graph respectively.  Consequently, randomized algorithms and graph preprocessing techniques have been utilized, strategically segmenting the graph into distinct components by manipulating bridges, articulation points, and degree-1 nodes~\cite{sariyuce2017graph}. We describe these key approaches  in Section~\ref{sec:related}.  In this paper, we explore the concept of removing degree-1 nodes in greater detail compared to~\cite{sariyuce2013shattering,sariyuce2017graph}. We analyze the impact of eliminating degree-1 vertices and derive a clear recurrence relation for calculating betweenness centralities from the 2-core.  This is crucial as it forms the basis for studying the impact of sampling on betweenness centrality computations and for understanding how to enhance BC estimation performance in real-world knowledge graph database systems.
Our empirical findings combined with our theoretical contributions show that on many real-world networks, even a single round of peeling degree-1 nodes, offers substantial advantages.

%% file: src/related.tex
\subsection{Betweenness centrality in practice} Betweenness centrality (BC) is extensively utilized in network analysis to assess the importance of nodes based on their strategic positions within the network. Here, we highlight several applications that illustrate the diverse utility and effectiveness of BC across various domains. Freeman~ \cite{freeman1977set} introduced it to identify
influential individuals within social networks who can facilitate efficient communication and information spread. These key individuals can be leveraged in strategies such as marketing or information dissemination within these networks.
Furthermore, in transportation networks, nodes with high betweenness centrality are crucial for traffic flow and often prioritized in infrastructure planning and maintenance \cite{freeman1977set}. These nodes are vital for ensuring efficient transportation and logistics operations. Newman~\cite{Newman2005} demonstrates that in communication networks, nodes with high betweenness centrality are crucial for efficient data routing. Similarly, in financial networks, institutions with high betweenness centrality are closely monitored for systemic risk, as their failure can have far-reaching impacts on the entire financial system. Betweenness centrality has been used to identify ``super-spreaders'' in networks in order to understand and control epidemics\cite{kitsak2010identification} and also perform connectivity attacks~\cite{mahmoody2016scalable},  Vogiatzis and Pardalos for evacuation crisis management~\cite{vogiatzis2016evacuation}, and Chen et al. to identify cricital bridges during seismic events~\cite{chen2023betweenness}.

\subsection{Brandes' algorithm} We discuss Brandes' algorithm in the context of unweighted graphs which is the main focus of this note and that yields an asymptotic improvement over the straightforward $O(n^3)$ all-pairs shortest paths (APSP) approach~\cite{clrs} for sparse graphs.  Brandes' algorithm breaks up the sum in definition~\ref{def:bc}  by the source node $s$. Specifically, Brandes' introduces the dependency of a source vertex $s\in V$ on a vertex $v\in V$ as 

\begin{equation}
\label{eq:delta}
\delta_{s}(v) = \sum_{\substack{t\in V \\ s\not= t\not= v}} \delta_{st}(v).
\end{equation}

We will conceptually represent the set of $\{\delta_s(v)\}$ values as an $n\times n$ matrix $\delta$ where $\delta[s,v] = \delta_s(v)$. Note that the betweenness centrality (BC)  of a node $v$ can be rewritten as

\begin{equation}
bc(v) = \frac{1}{(n-1)(n-2)} \sum_{s\neq v} \delta_{s}(v).
\end{equation}

\noindent In other words, the BC score of a node $v$ is simply the normalized sum of the respective column of $\delta$. Brandes observed that the BC scores can be computed using a two phase algorithm.  The first phase (top-down) computes the shortest path DAG routed at $s$ for each source node. If the graph is unweighted, we apply simple breadth first search (BFS) resulting in $O(n+m)$ time for each source $s$, otherwise we use Dijkstra's algorithm which takes $O(m+n \log n)$ time in the main memory (RAM) computational model~\cite{clrs}. 
The second phases goes bottom-up and updates the $\delta$ values.  Brandes provides a recursive way for the bottom-up phase. To describe it, we need to introduce some further notation following Brandes~\cite{brandes2001faster}.

\begin{definition}
Given a graph $G(V,E)$, the  predecessors  of a vertex $w\in V$ on a shortest path from $s$ to $w$ is a subset $P_s(v)\subseteq V$  such that 

$$ P_s(w) =\{v \in V :  d(s,w) = d(s,v) + 1, (v,w) \in E\},$$
where $d(s,t)$ denotes the shortest path distance between $s,t$.
\end{definition}

\noindent Brandes proved the following set of equations, which are valid for all \( s \) and \( v \), and these form the foundation of his algorithm.

\begin{equation}
\label{eq:deltaseq}
\delta_{s}(v) = \sum_{w \in V:  v \in P_s(w)} \frac{\sigma_{sv}}{\sigma_{sw}}(1+\delta_{s}(w)).
\end{equation}

The   run time for the first phase is $O(n(n+m))$ and for the second phase $O(nm)$ for an unweighted graph, yielding a total complexity $O(n^2+nm)$. For weighted graph, the total run time assuming Dijkstra's algorithm is $ O(n(m+n\log n)+nm) = O(nm+n^2\log n)$. Although Brandes' algorithm is the fastest for computing exact BCs, its time complexity can be impractical for many real-world applications, such as very large graphs or extensive datasets.

\subsection{Speeding up BC computation} We focus on the approaches that lie closest to our work; it is worth noting that a wide range of other methods for parallel and distributed betweenness centrality computations have been developed, e.g.,~\cite{bader2006parallel,hoang2019round}. A key method to accelerate the computation of betweenness centrality involves easing the requirement for precise calculations and employing randomized algorithms. One popular sampling approach~\cite{brandes2007centrality,hayashi2015fully,jacob2005algorithms} is to  sample $K$ nodes uniformly at random as sources for  Single Source Shortest Path (SSSP) computations and use this information to estimate the BC scores.   If $K$ is set equal to  \(O\left(\frac{1}{\epsilon^2} (\log n + \ln \frac{1}{\delta})\right)\), where \(\epsilon\) represents the approximation error, \(\delta\) the confidence level, we obtain an $(\epsilon, \delta)$ approximation with high probability for all nodes, namely $| \tilde{\mybc}(u) - \bc[]{u}| \leq \epsilon $ with probability at least $1-\delta$.  Riondato and Kornaropoulos suggested a method of sampling shortest paths and utilized the concept of VC dimension~\cite{vapnik2015uniform} to enhance the required sample size from \(\log n\) to \(\log \text{diam}(G)\), where \(\text{diam}(G)\) denotes the diameter of the graph. Geisberger et al.~\cite{geisberger2008contraction} prove that the Riondato-Kornaropoulus bound can lead to an overestimation.  Riondato and Upfal bypass the need for computing the graph diameter by analyzing the sampling of pairs of nodes using Radamacher averages~\cite{riondato2018abra},  see also~\cite{cousins2023bavarian}. Bader et al. proposed an adaptive sampling algorithm that effectively approximates the betweenness centrality scores for the most influential vertices, as detailed in~\cite{bader2007approximating}. Their bounds were further improved by Ji et al.~\cite{ji2016refining}. Pellegrina and Vandin focus also on providing a high-quality approximation of the top-k betweenness centralities~\cite{pellegrina2023silvan}. 
A different line of work~\cite{puzis2015topology,baglioni2012fast,sariyuce2013shattering,sariyuce2017graph} pre-processes the input graph by removing articulation points, bridges and degree-1 nodes. Notice that if a node $u$ has degree equal to 1, then the edge incident to it is a bridge. Sariyuce et al.  identify the bi-connected components of the graph and consolidating them into ``supernodes''. These supernodes are subsequently linked within the graph's biconnected tree. The crucial insight is that if a shortest path starts and ends at two distinct nodes within this tree, then all shortest paths connecting these nodes will pass through the same edges of the tree. This allows for the independent computation of betweenness centrality in each component, followed by an aggregation of the results. 
However, as pointed out by Riondato and Kornaropoulos~\cite{riondato2014fast} do not provide a complexity analysis (neither worst nor average case) of their algorithm. Naturally, they do not provide any justification of theoretical speedups.  Compared to Sariyuce et al.~\cite{sariyuce2013shattering,sariyuce2017graph}, we focus on the special case of degree-1 vertices, since the degree sequence is fast to compute and real-world networks tend to have many such nodes.  We derive a clear recurrence relation that illuminates the interplay between degree-1 nodes and BC scores, providing a theoretical foundation for potential speedups. To our knowledge, this is the first time our equations establish the basis for understanding the effect of peeling degree-1 nodes from the original graph on approximate BC estimation. Bentert et al.~\cite{bentert} propose an algorithm that combines 
the empirical works of~\cite{baglioni2012fast,sariyuce2013shattering,sariyuce2017graph,puzis2015topology} and a novel degree-2 vertex processing to yield an exact algorithm for BC that has worst time run time complexity  $O(n \cdot \textrm{OPT}_{fac})$ where $\textrm{OPT}_{fac}$ is the minimum feedback edge set of the input graph.

\subsection{Theoretical preliminaries: probabilistic inequalities}

We use the following two  probabilistic inequalities due to Hoeffding~\cite{hoeffding1994probability} and Bernstein~\cite{bernstein1924modification}, see also~\cite{vershynin2018high,boucheron2013concentration}.

\begin{proposition}[Hoeffding's inequality~\cite{hoeffding1994probability}]
\label{prop:hoeffding}

Let $X_1, \ldots, X_k$ be independent identically distributed (iid) random variables with $0 \leq X_i \leq M$  $(i = 1, \ldots, k)$ and an arbitrary \(\xi \geq 0\),

$$\Prob{ \left| \frac{X_1 + \cdots + X_k}{k} - E \left( \frac{X_1 + \cdots + X_k}{k} \right) \right| \geq \xi } \leq e^{-2k \left( \frac{\xi}{M} \right)^2}.
$$ 
\end{proposition}

\begin{proposition}[Bernstein's inequality for bounded distributions~\cite{bernstein1924modification}]
\label{prop:bernstein}

Let $X_1, \ldots, X_N$ be independent mean-zero random variables such that $|X_i| \leq K$ for all $i$. Then, for every $t \geq 0$, we have
\[
\mathbb{P}\left(|\sum_{i=1}^N X_i| \geq t\right) \leq 2\exp\left(-\frac{t^2/2}{\sigma^2 + Kt/3}\right).
\]
Here $\sigma^2 = \sum_{i=1}^N \mathbb{E}X_i^2$ is the variance of the sum.

\end{proposition}


%% file: src/proposed.tex
\subsection{Recursive BC computation from the 2-core}

 Let \(G_0 = G, G_1, G_2, \ldots\) be a nested sequence of graphs such that \(G_i \supseteq G_{i+1}\) for all \(i \geq 0\), obtained by iteratively removing the degree 1 nodes.  
Let \(G_{i^\star}\) denote the 2-core of \(G\) (possibly empty). We denote these sets as \(V_1^{(0)}, V_1^{(1)}, \ldots\), where \(V_1^{(i)}\) is the set of degree 1 nodes in \(G_i\) that are removed to obtain \(G_{i+1}\). Let \(V_{\geq 2}^{(i)}\) be the set of nodes in \(G_i\) with degree at least 2. Define \(Y^{(i)}\) as the set of neighbors of degree 1 nodes in \(G_i\). For each \(u \in Y^{(i)}\), let \(\degree{u}{i}\) represent the number of its neighbors whose degree is 1. Denote by \(N^{(i)}(u)\) the set of neighbors of node \(u \in G_i\). Finally, let \(n_i = |V^{(i)}|, m_i = |E^{(i)}|\) be the number of nodes and edges in \(G_i\) for all $i$ respectively.

\begin{thm}
\label{thm:sigma}
Let $\sigma_{st}^{(i)}$ be the number of shortest paths between $s,t$ in $G_{i}$. Then the following recursive equations hold for all $s,t \in V, i\in \{0,\ldots, i^\star-1\}$ holds:  

\begin{align*}
\sigma_{st}^{(i)}&= 
\begin{cases}
\sigma_{st}^{(i+1)}, & \text{for all~~} s,t \in V_{\geq 2}^{(i)}  \\
\sigma_{sy}^{(i+1)}, & \text{for all~~} s \in V_{\geq 2}^{(i)}, t \in V_1^{(i)}  \text{~~where~~}  (y,t) \in G_{i }   \\
\sigma_{yt}^{(i+1)}, & \text{for all~~} s \in V_1^{(i)}, t \in V_{\geq 2}^{(i)} \text{~~where~~}  (s,y) \in G_{i} \\
\sigma_{yy'}^{(i+1)}, & \text{for all~~}   s,t \in V_1^{(i)} \text{~~where~~} (s,y), (y',t) \in G_{i}  
\end{cases}
\end{align*}
\end{thm}

\begin{proof}
The proof   relies on the fundamental property that the only way to reach or exit a degree 1 vertex is through its unique neighbor. We consider the following cases. 

\noindent \underline{Case I: $s,t \in V_{\geq 2}^{(i)}$.}
Notice that if $s,t \in V_{\geq 2}^{(i)}$ they also appear in $G_{i+1}$. The sets of shortest paths between $s,t$ in $G_i, G_{i+1}$ are identical, as no $s-t$ shortest paths in $G_i$ go through nodes from $V_1^{(i)}$, but only through nodes in $V_{\geq 2}^{(i)}$.

\noindent  \underline{Case II: $s \in V_{\geq 2}^{(i)}, t \in V_1^{(i)}$.} Let \(y \in V_{\geq 2}^{(i)}\) be the unique neighbor of \(t\) in \(G_i\). Observe that each shortest path from \(s\) to \(y\) corresponds to a shortest path from \(s\) to \(t\), and no other such paths can exist. Consequently, we have \(\sigma_{st}^{(i)} = \sigma_{sy}^{(i+1)}\).

\noindent  \underline{Case III: $s \in V_1^{(i)}, t \in V_{\geq 2}^{(i)}$.} Arguing  similarly to case II,  \(\sigma_{st}^{(i)} = \sigma_{yt}^{(i+1)}\).

\noindent 
\underline{Case IV: $s,t \in  V_1^{(i)}$.} There exists a bijection between the shortest paths from $s$ to $t$ in $G_i$ and the shortest paths from $y$ to $y'$ in $G_{i+1}$ where $y,y'$ are the neighbors of $s,t$ in $G_i$ by appending the edges $(s,y)$ and $(y',t)$ to each $y-y'$ shortest path. 
\end{proof}

\begin{cor}
\label{cor:sigmastu}
Let $\sigma_{st}^{(i)}(u)$ be the number of shortest paths between $s,t$ in $G_{i}$ that pass through $u \in V_{\geq 2}^{(i)}$. Then, the following holds: 

\begin{align*}
\sigma_{st}^{(i)}(u)&= 
\begin{cases}
\sigma_{st}^{(i+1)}(u), & \text{for all~~} s,t \in V_{\geq 2}^{(i)}  \\
\sigma_{sy}^{(i+1)}(u), & \text{for all~~} s \in V_{\geq 2}^{(i)}, t \in V_1^{(i)}  \text{~~where~~}  (y,t) \in G_{i }   \\
\sigma_{yt}^{(i+1)}(u), & \text{for all~~} s \in V_1^{(i)}, t \in V_{\geq 2}^{(i)} \text{~~where~~}  (s,y) \in G_{i} \\
\sigma_{yy'}^{(i+1)}(u), & \text{for all~~}   s,t \in V_1^{(i)} \text{~~where~~} (s,y), (y',t) \in G_{i}  
\end{cases}
 \end{align*}
\end{cor}

We omit the proof of the corollary as it is identical to Theorem~\ref{thm:sigma} and establishes the same type of bijection between   two  respective sets of shortest paths in $G_i, G_{i+1}$. Notice that  we consider $u \in V_{\geq 2}^{(i)}$, since if $u \in V_1^{(i)}$ the total number of shortest paths that pass through it is zero in $G_i$. Having established Theorem~\ref{thm:sigma} and Corollary~\ref{cor:sigmastu} we proceed to the following key theorem that establishes a connection between the betweenness centrality scores \(\bc[i]{u}\) and \(\bc[i+1]{u}\) for graphs \(G_i\) and \(G_{i+1}\), respectively.
 
\begin{thm} 
\label{thm:rec}
Let $\bc[i]{u}$ be the betweenness centrality of node $u$ in $G_i$ for $i=0,\ldots,i^\star$ such that $G_{i^\star}$ is the maximal 2-core of $G$. Then, the following recursive equations hold for all $i$:

\begin{empheq}[box=\fbox]{align*}
\bc[i]{u}  (n_i - 1) (n_i-2) &= \bc[i+1]{u}  (n_{i+1} - 1) (n_{i+1}-2)+   \\
&  \degree{u}{i} ( \degree{u}{i} -1) + \\ 
& 2\cdot  \degree{u}{i}  \cdot \big(n_i - ( \degree{u}{i} +1) \big) +  \\
&  2 \cdot \sum_{ \substack{y \in V_{\geq 2}^{(i)} \\ y\neq u} } \sum_{\substack{y' \in Y^{(i)} \\ y' \neq y \neq u}}  \degree{y}{i}    \frac{\sigma_{yy'}^{(i+1)}(u)}{\sigma_{yy'}^{(i+1)}} + \\
&  \sum_{ \substack{y \in Y^{(i)} \\ y\neq u} } \sum_{\substack{y' \in Y^{(i)} \\ y' \neq y \neq u}}  \degree{y}{i} \cdot \degree{y'}{i} 
 \cdot \frac{\sigma_{yy'}^{(i+1)}(u)}{\sigma_{yy'}^{(i+1)}} \\ 
& \forall u \in V, i \in \{0,\ldots,i^\star-1\}
\end{empheq}

\end{thm}

\begin{proof}
Consider an arbitrary index $i \in \{0,\ldots,i^\star-1\}$ and the two graphs $G_i, G_{i+1}$ which differ by the vertex set $V_1^{(i)}$ removed from $G_i$. By the definition~\eqref{def:bc}, notice that the quantity $\bc[i]{u}  (n_i - 1) (n_i-2) $ is equal to

\begin{align*}
 \bc[i]{u}  (n_i - 1) (n_i-2) &=  
\sum_{\substack{s, t \in V^{(i)} \\ s \not= t \not= u}} \delta_{st}^{(i)}(u) =
 \sum_{\substack{s, t \in V_{\geq 2}^{(i+1)} \\ s \not= t \not= u}} \delta_{st}^{(i+1)}(u)
 + Z  \\
 & = \bc[i+1]{u}  (n_{i+1} - 1) (n_{i+1}-2)+Z,
\end{align*}

\noindent where $Z$ is the term representing the extra contributions to $\bc[i]{u}  (n_i - 1) (n_i-2)$ from paths involving nodes from the set of  $V_1^{(i)}$ that does not exist in $G_{i+1}$.  We break down the contributions to this term by considering all possible cases of sources and targets.

\underline{Case I: $s,t \in  N_1^{(i)}(u)$} Let us consider the scenario when both the source and the destination nodes \(s\) and \(t\) are degree 1 nodes that are neighbors of $u$ and were removed from \(G_i\), i.e., \(s, t \in N(u) \cap V_1^{(i)} = N_1^{(i)}(u)\). In this case, there exist \(\degree{u}{i} \cdot (\degree{u}{i} - 1)\) pairs of such nodes, each pair being connected via a unique shortest path passing through \(u\). Consequently, this set of pairs contributes \(\degree{u}{i} \cdot (\degree{u}{i} - 1)\) to \(Z\).

\underline{Cases II(a), II(b): \( s \in N_1^{(i)}(u) \) and \( t \in V^{(i)} \setminus N_1^{(i)}(u) \), \textit{or vice versa}.} In this scenario   there exists \\ $\degree{u}{i} $ choices for the source \(s \in N_1^{(i)}(u)\) and $\big(n_i - (\degree{u}{i} + 1)\big)$ choices for the destination node \\ \(t \in V^{(i)} \setminus \big(u \cup N_1^{(i)}(u)\big)\), yielding in total  $$ \degree{u}{i} \cdot \big(n_i - (\degree{u}{i} + 1)\big) $$ 

\noindent possible source-target pairs \(s, t\). For each such pair, the pair-dependency is equal to 1 since all shortest paths from \(s\) to \(t\) in $G_i$ pass through \(u\). The symmetric case where \(t \in N_1^{(i)}(u)\) and \(s \in V^{(i)} \setminus \big(u \cup N_1^{(i)}(u)\big)\) yields an equal term. Therefore, the total contribution is $ 2 \cdot \degree{u}{i} \cdot \big(n_i - (\degree{u}{i} + 1)\big). $

\underline{Cases III(a) and III(b): $s \in V_{\geq 2}^{(i)}$, $t \in V_1^{(i)} \backslash N_1^{(i)}(u)$, \textit{or vice versa}.  } Consider the case $s \in V_{\geq 2}^{(i)}$ and $t \in V_1^{(i)} \backslash N_1^{(i)}(u)$.  By applying Theorem~\ref{thm:sigma} we obtain that  contribution of such pairs is equal to

\begin{align*}
\sum_{ \substack{y \in V_{\geq 2}^{(i)} \\ y\neq u} } \sum_{t \in V_1^{(i)} \backslash N_1^{(i)}(u)} \frac{\sigma_{yt}^{(i)}(u)}{\sigma_{yt}^{(i)}} &=\sum_{ \substack{y \in V_{\geq 2}^{(i)} \\ y\neq u} } \sum_{\substack{y' \in Y^{(i)} \\ y' \neq y \neq u}}  \degree{y}{i}    \frac{\sigma_{yy'}^{(i+1)}(u)}{\sigma_{yy'}^{(i+1)}},
\end{align*}

\noindent where $Y^{(i)}$ is the set of nodes in $G_i$ that are adjacent to $V_1^{(i)}$.  Similarly, by considering the case $s \in V_1^{(i)} \backslash N_1^{(i)}(u)$, $t \in V_{\geq 2}^{(i)}$ we obtain another equal contribution term. This leads to the overall contribution term  $ 2 \cdot \sum_{ \substack{y \in V_{\geq 2}^{(i)} \\ y\neq u} } \sum_{\substack{y' \in Y^{(i)} \\ y' \neq y \neq u}}  \degree{y}{i}    \frac{\sigma_{yy'}^{(i+1)}(u)}{\sigma_{yy'}^{(i+1)}}$.

\underline{Case IV: $s,t \in V_1^{(i)}\backslash N_1^{(i)}(u)$. } Consider $s,t \in V_1^{(i)}\backslash N_1^{(i)}(u)$.  Let $y, y'$ be the neighbors of $s,t$ respectively. Notice that $y \neq y'$ otherwise the shortest path from $s,t$ does not pass through $u$ but through $y$.  By rearranging the   sum of the contribution of these pairs, applying Theorem~\ref{thm:sigma}, grouping according to the set $Y^{(i)}$ and observing that for each pair $y,y' \in V_{\geq 2}^{(i)}$ there exist $\degree{y}{i} \cdot \degree{y'}{i} $ terms equal to 
$\frac{\sigma_{yy'}^{(i+1)}(u)}{\sigma_{yy'}^{(i+1)}}$ 
we obtain:

$$ \sum_{s \in V_1^{(i)}\backslash N_1^{(i)}(u)
} \sum_{t \in V_1^{(i)}\backslash N_1^{(i)}(u)} \frac{\sigma_{st}^{(i)}(u)}{\sigma_{st}^{(i)}} = \sum_{ \substack{y \in Y^{(i)} \\ y\neq u} } \sum_{\substack{y' \in Y^{(i)} \\ y' \neq y \neq u}}  \degree{y}{i} \cdot \degree{y'}{i} 
 \cdot \frac{\sigma_{yy'}^{(i+1)}(u)}{\sigma_{yy'}^{(i+1)}}.$$

\noindent Summing the contributions from these cases results in the theorem statement.

\end{proof}

\noindent Theorem~\ref{thm:rec} naturally implies an algorithm that iteratively peels the degree 1 nodes until it is no longer possible. 
Suppose that we compute  $\sigma^{(i^\star)}_{st}, \sigma_{st}^{(i^\star)}(u)$ in  $G_{i^\star}$ for all $s,t,u \in G_{i^\star}$. Notice that this information allows us to compute the betweenness centrality scores $\bc[i^\star]{u}$ for all $u \in G_{i^\star}$. Then, we apply recursively from index $i^\star-1$ down to $i=0$ Theorems~\ref{thm:sigma}, \ref{thm:rec} and Corollary~\ref{cor:sigmastu} to compute the betweenness centrality scores in graphs $G_{i^\star-1}, \ldots,G_0$.   A straight-forward implementation, performs an APSP computation in the 2-core and then  updates the $\sigma$-quantities using Theorem~\ref{thm:sigma} and Corollary~\ref{cor:sigmastu} from $i+1$ to $i$ ($i=1,\ldots,i^\star-1$). Notice we need to update $O(|V_1^{(i)}|^2+|V_{\geq 2}^{(i)}||V_1^{(i)}|)$ quantities in constant time.   Given the algorithm's heightened sensitivity $O(t_{APSP}(G_{i^\star})+\sum_{i=1}^{i^\star-1} O(|V_1^{(i)}|^2+|V_{\geq 2}^{(i)}||V_1^{(i)}|)$ to the properties of its inputs, this nuanced expression is anticipated.


     


\subsection{ BC from 1-round of peeling}

 We will focus on a single iteration of peeling degree-1 nodes and explore an effective implementation strategy.   There is an important practical reason why we focus  on this case, besides obtaining a less nuanced time complexity: most of the iterative removal of degree-1 nodes across various networks typically occurs during the initial round.  Specifically, consider the  set of social, collaboration and communication networks publicly available from SNAP~\cite{snap} described in the Appendix~\ref{sec:appendix}.  On average, the number of peeling rounds is 3.89 with a standard deviation of 2.08, as detailed in the histogram in Figure~\ref{fig:combined}(a). Despite the relatively small average number of peeling rounds, the first round consistently results in the largest reduction in the number of nodes across all networks, as shown in Figure~\ref{fig:combined}(b). For datasets with multiple peeling rounds, the ratio of $\frac{V_1^{(0)}}{V_1^{(1)}}$ has an average of 67.82, a standard deviation of 113.56, and maximum and minimum values of 446.53 and 9.4, respectively.

\begin{figure}[h!]
    \centering
    \begin{tabular}{cc}
        \includegraphics[width=0.34\textwidth]{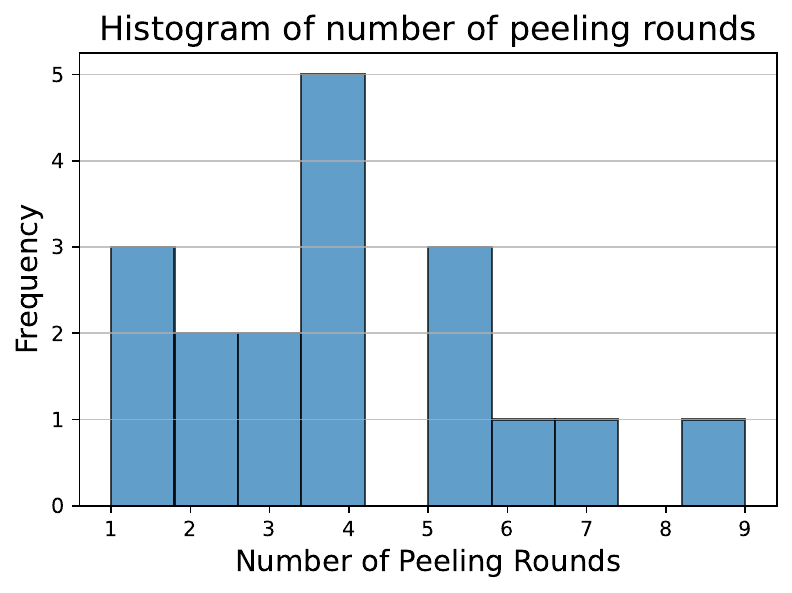} &
        \includegraphics[width=0.55\textwidth]{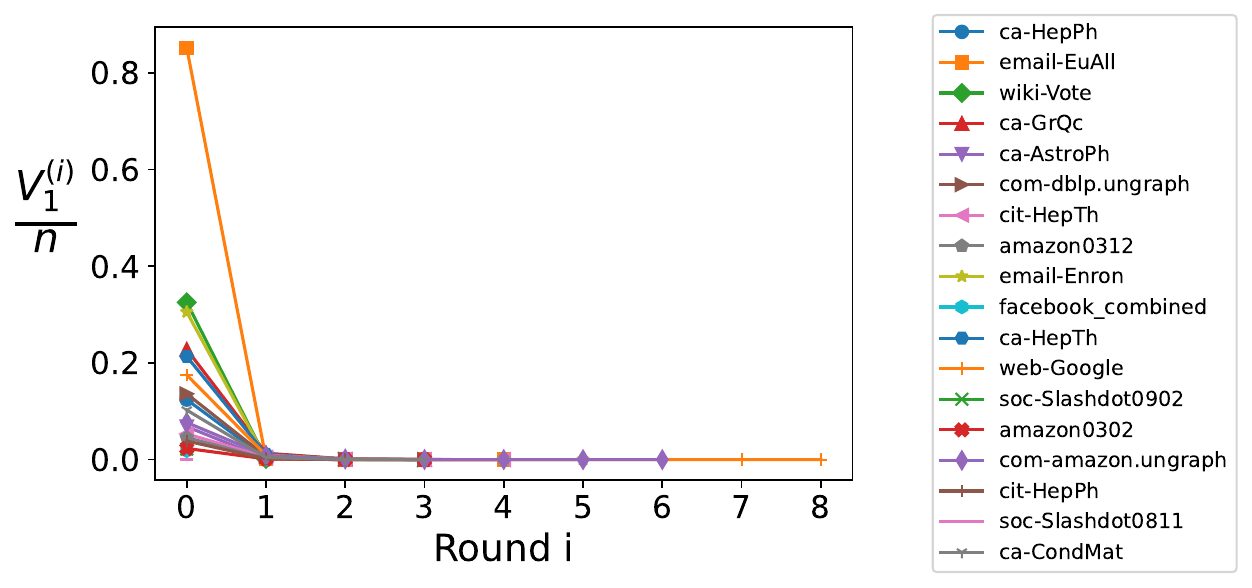} 
        \\
    \end{tabular}
    \caption{   (a) Histogram of the number of peeling rounds.  
          (b) Fraction of degree-1 nodes $\frac{V_1^{(i)}}{n}$ removed in each round $i$ for each dataset.}
    \label{fig:combined}
\end{figure}

\noindent Algorithm~\ref{alg:oneround} illustrates the adaptation of Brandes' algorithm to handle the removal of $V_1^{(0)} := V_1$.   We refer to the remaining set of vertices as $V_{\geq 2}$, and the graph induced by $V_{\geq 2}$ is denoted as $\tilde{G}=G^{(1)}$.
Our algorithm introduces in addition to the $\delta$ values from which the BC scores can be directly derived, a set of  $\zeta$ values.  In particular, for any $s,u \in V_{\geq 2}$ we define 

$$ \zeta_s(u) = \sum_{\substack{t \in V_{\geq 2} \\ s \neq t \neq u  }} deg_1(t) \cdot \frac{\sigma_{st}^{(1)}(u)}{\sigma_{st}^{(1)}}.$$

Our algorithm computes in lines 8-11 the $\delta,\zeta$ quantities in $\tilde{G}$, by invoking the function $\_${\sc Accumulate}. Then it considers the original input graph, by updating the $\delta$ values appropriately to capture the effect of the removed set $V_1$. We prove the correctness of $\_${\sc Accumulate} in the next theorem and the algorithm correctness in Theorem~\ref{thm:correctness}.

 \begin{thm}
 \label{thm:reczeta}
The $\zeta$ values satisfy the following recursive equation: 
\[
\zeta_{s}(u) = \sum_{w: u \in P_s(w)} 
\frac{\sigma_{sv}}{\sigma_{sw}} \cdot (deg_1(w) + \zeta_{s}(w)).
\]
\end{thm}

\begin{proof}
 
Let $\sigma_{st}(u, e)$ be the number of shortest paths from $s$ to $t$ that contain both node $u$ and edge $e$. Brandes defined a convenient  extension the  pair-dependency~\cite{brandes2001faster} $\delta_{st}(u, \{u, w\})=\frac{\sigma_{st}(u, e)}{\sigma_{st}}$, that includes an edge $e \in E$ and proved  that


\[
\delta_{st}(u, \{u, w\}) = \begin{cases}
\frac{\sigma_{su}}{\sigma_{sw}} & \text{if } t = w \\
\frac{\sigma_{su}}{\sigma_{sw}} \cdot \frac{\sigma_{st}(u)}{\sigma_{st}} & \text{if } t \neq w
\end{cases}
\]

\noindent 
With the above, we  obtain the following expression: 

\begin{align*}
\zeta_{s}(u) &= \sum_{t \in V_{\geq 2}} deg_1(t) \frac{\sigma_{st}(u)}{\sigma_{st}}  = \sum_{t \in V_{\geq 2}} 
\sum_{\substack{w \in V_{\geq 2}:\\ u \in P_s(w)}}
\deg_1(t) \cdot \frac{\sigma_{st}(u, \{u,w\})}{\sigma_{st}}  \\  
&=\sum_{\substack{w \in V_{\geq 2}:\\ u \in P_s(w)}}
\sum_{t \in V_{\geq 2}}  deg_1(t) \cdot \frac{\sigma_{st}(u, \{u,w\})}{\sigma_{st}}=\sum_{\substack{w \in V_{\geq 2}:\\ u \in P_s(w)}}
\sum_{t \in V_{\geq 2}}    \zeta_{st}(u, \{u, w\}) \\
&= \sum_{\substack{w \in V_{\geq 2}:\\ u \in P_s(w)}} \Big(  deg_1(w) \frac{\sigma_{su}}{\sigma_{sw}} +  \sum_{t \in V_{\geq 2} \setminus \{w\}} deg_1(t) \frac{\sigma_{su}}{\sigma_{sw}} \cdot \frac{\sigma_{st}(w)}{\sigma_{st}}   \Big) \\ 
&= \sum_{\substack{w \in V_{\geq 2}:\\ u \in P_s(w)}}  \frac{\sigma_{su}}{\sigma_{sw}} \Big(  deg_1(w) + \zeta_s(w) \Big).
\end{align*}

\end{proof}

\noindent We present the key corollary of Theorem~\ref{thm:reczeta} and Theorem 6 from~\cite{brandes2001faster}.

\begin{cor}
Function $\_${\sc Accumulate}   computes the $\delta, \zeta$ quantities in $\tilde{G}$ in $O(\tilde{n}\tilde{m})$ time, where $\tilde{n}$ and $\tilde{m}$ are the number of nodes and edges respectively in $\tilde{G}$.
\end{cor}

\noindent We are now prepared to present the correctness proof of our proposed algorithm. We will show that the $\delta^{(1)}$ and $\zeta$ values from $\tilde{G}$ are correctly utilized to update the $\delta^{(0)}$ values in the original graph $G$. A key observation is that $\delta_s^{(1)}(u)=0$ for all $u \in V_1$. This indicates that the columns of the $\delta$ matrix are entirely zeros, meaning they do not need to be updated or even stored, saving space. Notice that even if we are updating only   columns of the $\delta$ matrix corresponding   to nodes $u \in V_{\geq 2}$ (line 13) we are potentially updating all rows (line 12).

\begin{thm}
\label{thm:correctness}
Algorithm~\ref{alg:oneround} computes exactly the betweenness centrality scores in the input graph $G$.  
\end{thm}

\begin{proof} 

We break up the sum of $\delta_s^{(0)}(u)$ into two terms, depending on whether the destination node $t$ is in $V_{\geq 2}$, $V_1 \backslash N(u)$ and $V_1 \cap N(u)$.  The correctness proof considers two cases, whether $s \in V_{\geq 2}$ or $s \in V_1$ in order to analyze the sum of the latter first two terms

\begin{align}
\delta_s^{(0)}(u) &= \sum_{t \neq {s,u}} \frac{\sigma_{st}^{(0)}(u)}{\sigma_{st}^{(0)}} = \sum_{\substack{t \neq {s,u} \\ t \in V_{\geq 2}}} \frac{\sigma_{st}^{(0)}(u)}{\sigma_{st}^{(0)}}
+  \sum_{\substack{t \neq {s,u} \\ t \in V_1 \backslash N(u)}} \frac{\sigma_{st}^{(0)}(u)}{\sigma_{st}^{(0)}} +\sum_{\substack{t \neq {s,u} \\ t \in V_1 \cap N(u)}} \frac{\sigma_{st}^{(0)}(u)}{\sigma_{st}^{(0)}} \label{eq:deltasum}
\end{align}

The correctness proof examines two scenarios: $s \in V_{\geq 2}$ and $s \in V_1$, to analyze the sum of the first two terms. The third term in equation~\eqref{eq:deltasum} equals $deg_1(u)$ in both cases, as accounted for in lines 25-28. Note that in an implementation, matrix $\eta$ does not need to be materialized; it is included only for clarity.

\underline{Case I: $s \in V_{\geq 2}$.} By Theorem~\ref{thm:sigma} and Corollary~\ref{cor:sigmastu}, since   $s, t \in V_{\geq 2}$, $\sigma_{st}^{(0)}(u)=\sigma_{st}^{(1)}(u), \sigma_{st}^{(0)}=\sigma_{st}^{(1)}$. Therefore the first term of the sum can be rewritten as 

$$ \sum_{\substack{t \neq {s,u} \\ t \in V_{\geq 2}}} \frac{\sigma_{st}^{(0)}(u)}{\sigma_{st}^{(0)}} = 
\sum_{\substack{t \neq {s,u} \\ t \in V_{\geq 2}}} \frac{\sigma_{st}^{(1)}(u)}{\sigma_{st}^{(1)}} = \delta_s^{(1)}(u).
$$

By invoking Theorem~\ref{thm:sigma} and Corollary~\ref{cor:sigmastu}, we group  the second summation term   according to the unique neighbor $y \in V_{\geq 2} \backslash u$ of each $t \in V_1$ to obtain the following expression: 

$$\sum_{\substack{t \neq {s,u} \\ t \in V_1\backslash N(u)}} \frac{\sigma_{st}^{(0)}(u)}{\sigma_{st}^{(0)}} = \sum_{y \in V_{\geq 2} \backslash u} deg_1(y) \cdot \frac{\sigma_{sy}^{(1)}(u)}{\sigma_{sy}^{(1)}}= \zeta_s(u).$$

\underline{Case II(A): $s \in V_1, s \notin N(u)$.}  Let $y$ be the neighbor of $s$ in $V_{\geq 2}.$ Similarly as in Case I, by using again Theorem~\ref{thm:sigma} and Corollary~\ref{cor:sigmastu}, we obtain that the first term in Equation~\eqref{eq:deltasum} is equal to $\delta_y^{(1)}(u)$ and the second term $\zeta_y(u)$.

\underline{Case II(B): $s \in V_1, s \in N(u)$.} In this case, all paths from $s$ to the rest of the nodes in $G$ go through $u$, and therefore $\delta_s^{(0)}(u)=n-2$. This is implemented in lines 14-15 in Algorithm~\ref{alg:oneround}.

Combining these cases establishes the correctness of the algorithm.

\end{proof}

 \begin{algorithm}
\caption{\label{alg:oneround} Calculate betweenness centrality and full influence information using one round peeling}
\begin{algorithmic}[1]
   \Function{BC\_One\_Round\_Peeling\_Full\_Info}{$G, k=None$}

    \State $V_1 \leftarrow \{ v \in V : \deg(v) = 1 \}$ 
    \Comment Set of degree 1 nodes 
    \State $V_{\geq 2} = V \backslash V_1$
    \Comment Nodes of degree 2 or higher
    
    \State $\mathrm{bc}[u] \gets 0.0$ $\forall u \in V$
    \Comment Initializations
    \State   $\delta \gets 0^{n \times n}$,  $\zeta \gets 0^{n \times n}$,
     $\eta  \gets 0^{n \times n}$   
     \State $deg_1[u] \gets |N(u) \cap V_1|$
     \Comment Number of degree one neighbors
    \State $Y \gets \{u : deg_1[u] \geq 1\}$
    \Comment Nodes with at least 1 degree neighbor, wlog $Y \subseteq V_{\geq 2}$
    
    \State $\tilde{G} \gets G \backslash V_1$
    \State flag $\gets$ \textbf{False} 
    \Comment Remove degree one nodes 
     \If{$k = $ \textbf{None} or $k \geq \tilde{n}$}
        \State $nodes \gets V_{\geq 2}$
    \Else
        \State flag $\gets$ \textbf{True}  
        \State $nodes \gets \text{random sample}( V_{\geq 2}, k)$
    \EndIf
    \For{$s \in nodes$}
        \State $S, P, \sigma, D \gets \Call{SSSP}{\tilde{G}, s}$
       \Comment{Solve BC on $\tilde{G}$}
        \State $\Call{\_Accumulate}{S, P, \sigma, s, \delta, \zeta, deg_1}$
    \EndFor
     \If{$\text{flag}$}
    \Comment Rescaling
    \For{$u \in V_{\geq 2}$}
        \State $\delta  \gets  \frac{\tilde{n}}{k} \times \delta$,  $\zeta  \gets  \frac{\tilde{n}}{k} \times \zeta$
    \EndFor
    \EndIf
    \For{$s \in V$} 
    \Comment{Update $n \times (n-|V_1|)$ $\delta$ values. Observe that if $u \in V_1$, then $\delta[s, u] \gets 0.0$.}
        \For{$u \in V_{\geq 2}$}
            \If{$s \in V_1  \cap N(u)$}
                \State $\delta[s, u] \gets (n - 2)$ 
            \ElsIf{$s \in V_1 \text{ and }   s \notin N(u)$}
                \State $\{ y\} \gets N(s)$
                \State $\delta[s, u] \gets \delta[y, u] + \zeta[y, u]$
            \ElsIf{$s \in V_{\geq 2} $}
                \State $\delta[s, u] \gets \delta[s, u] + \zeta[s, u]$
            \EndIf
        \EndFor
    \EndFor
    \For{$u \in Y$}
        \For{$s \in V \backslash ( \{u\} \cup (N(u) \cap V_1))$}
            \State $\eta[s, u] \gets deg_1[u]$
        \EndFor
    \EndFor
      
    \State $\delta \gets \delta + \eta   $

    \For{$u \in V_{\geq 2}$}      
            \State $\mathrm{bc}[u]=\frac{1}{(n-1)(n-2)} \sum_{s \in V} \delta[s,u]$ 
    \EndFor  
    \State {\bf return} $\mathrm{bc}$
\EndFunction
\end{algorithmic}
\end{algorithm}

\begin{algorithm}
\caption{}
\begin{algorithmic}[1]
\Function{$\_$Accumulate}{$S, P, \sigma, s, \delta, \zeta, \text{deg}_1$}
    \State 
    $\delta_{\text{tmp}}[u], \zeta_{\text{tmp}}[u]  \gets 0.0$ for all $u \in V$
    \While{$S \neq \emptyset$}
        \State $w \gets S.\text{pop}()$
        \State $\text{coeff}_{\sigma} \gets \frac{1 + \delta_{\text{tmp}}[w]}{\sigma[w]}$
        \State $\text{coeff}_{\zeta} \gets \frac{\text{deg}_1[w] + \zeta_{\text{tmp}}[w]}{\sigma[w]}$
        \For{$v \in P[w]$}
            \If{$v \neq s$}
                 \State $\delta_{\text{tmp}}[v] \gets \delta_{\text{tmp}}[v] +  \text{coeff}_{\sigma} 
 \times \sigma[v] $
                \State $\zeta_{\text{tmp}}[v] \gets \zeta_{\text{tmp}}[v] + \text{coeff}_\zeta \times \sigma[v]  $
            \EndIf
        \EndFor
        \State $\delta[s][w] \gets \delta_{\text{tmp}}[w]$
        \State $\zeta[s][w] \gets \zeta_{\text{tmp}}[w]$
    \EndWhile
\EndFunction
\end{algorithmic}
\end{algorithm}

\paragraph{{\bf Time complexity.}} We analyze the time complexity of an efficient implementation of Algorithm~\ref{alg:oneround}, showing also that it can provably yield asymptotic speedups over Brandes' $O(nm)$ implementation. Lines 8-10 require $O(\tilde n \tilde{m})$ time, using Brandes' algorithm on the pruned graph $\tilde{G}$. Lines 12-23 require  $O(n \cdot \tilde{n})$ time, as there are at most $O(n \cdot \tilde{n})$ entries in the $\delta$ matrix that need to be updated and each update (i.e., lines 15, 18, 20) requires constant time.  Finally lines 24-28 require $O(|Y|n)$ times where $Y$ is the set of nodes in $G$ that have at least one degree-1 neighbor. Combining these terms we obtain that the time complexity of Algorithm~\ref{alg:oneround} in the standard RAM model is $O(\tilde n \tilde{m} + n \cdot \tilde{n} + |Y| n)$. While this expression appears still nuanced, it is clean and we can justify  asymptotic speedups. For example, in the case of a star graph where $Y=\{\text{center node}\}, \tilde{n}=1, \tilde{m}=0$ we obtain runtime complexity $O(n)$ instead of $O(n^2)$ using Brandes' algorithm on the star graph. In general for trees, peeling will always result in an empty core and peeling can justify $O(n)$ speedups, a fact pointed out in~\cite{sariyuce2013shattering} as well.

A more interesting example is a simplified variant of the core-periphery model that is a mathematical model with a long history~\cite{nemeth1985international,avin2018elites,rombach2017core,zhang2015identification}. The core-periphery model is often used to describe social networks, but it is used to describe other types of networks. In the core-periphery model, the network is divided into two parts: the core and the periphery.  The core is made up of the most important nodes in the network, and the periphery is made up of the least important nodes. The core and the periphery are connected by a set of links.   There exists numerous theoretical works that under reasonable assumptions that mimick real-world properties of networks show that the core has sublinear size~\cite{papachristou2021sublinear}.  Let's consider a core-periphery graph where the core consists of a clique with $O(\sqrt{n})$ nodes, each connected to $O(\sqrt n)$ degree-1 nodes which form the periphery. The total number of edges is $O(n)$.  The overall complexity is $O(n^{3/2})$ which is significantly less than Brandes' run time $O(nm)=O(n^2)$.   

\subsection{Memory efficient implementation}

In many practical scenarios, the complete information stored in the matrix $\delta$ is unnecessary. Moreover, for extensive networks, storing the $\delta$ matrix can be infeasible. Therefore, we can simplify our algorithm to compute only the betweenness centralities of the nodes, as demonstrated in Algorithm~\ref{alg:memefficient}. Our pseudocode can be incorporated with minor modifications into a standard betweenness centrality implementation, such as the one provided in NetworkX~\cite{hagberg2020networkx}. The proof of correctness remains the same as for Algorithm~\ref{alg:oneround}. As a small remark for perhaps the less obvious line 26,  the term is derived from combining the terms $deg_1(u) \times (n-2) + deg_1(u)  \times (n-(deg_1(u) +1))=d \times (2n-deg_1(u) -3)$.

\input{src/mem_efficient}







%% file: src/mem_efficient.tex
\begin{algorithm}
\caption{\label{alg:memefficient} Calculate Betweenness Centrality  using One Round Peeling (Memory Efficient)}
\begin{algorithmic}[1]
\Function{BC\_One\_Round\_Peeling\_Mem\_Efficient}{$G, k=$\textbf{None}} \\
\Comment{$k$ is the sample size. When $k$={\bf None}, we perform the exact computation}
    \State $V_1 \leftarrow \{ v \in V : \deg(v) = 1 \}$      
    \State $V_{\geq 2} = V \backslash V_1$, $\tilde{n} \gets |V_{\geq 2}|$ 
    \Comment Nodes of degree 2 or higher

    \State $\mathrm{bc}[u] \gets 0.0$ $\forall u \in V$
    \Comment Initializations
     \State $deg_1[u] \gets |N(u) \cap V_1|$ $\forall u \in V$
     \Comment Number of degree one neighbors
    \State $Y \gets \{u : deg_1[u] \geq 1\}$
    \Comment Nodes with at least 1 degree neighbor, wlog $Y \subseteq V_{\geq 2}$
    \State $\tilde{G} \gets G \backslash V_1$
    \State flag $\gets$ \textbf{False} 
    \If{$k = $ \textbf{None} or $k \geq \tilde{n}$}
        \State $nodes \gets V_{\geq 2}$
    \Else
        \State flag $\gets$ \textbf{True}  
        \State $nodes \gets \text{random sample}( V_{\geq 2}, k)$
    \EndIf
    \For{$s \in nodes$}
        \State $S, P, \sigma, D \gets \Call{SSSP}{\tilde{G}, s}$
        \State $\mathrm{bc} \gets \Call{\_Accumulate\_Basic\_Mem\_Efficient}{\mathrm{bc}, S, P, \sigma, s, \text{deg}_1}$
    \EndFor
    \If{$\text{flag}$}
    \Comment Rescaling
    \For{$u \in V_{\geq 2}$}
        \State $\mathrm{bc}[u] \gets  \frac{\tilde{n}}{k} \times \mathrm{bc}[u]$ for all $u \in V_{\geq 2}$
    \EndFor
    \EndIf
    \For{$u \in V_{\geq 2}$}
        \State $\mathrm{bc}[u] \gets \mathrm{bc}[u] + (2n - 3 - \text{deg}_1[u]) \times \text{deg}_1[u]$
        \State $\mathrm{bc}[u] \gets \mathrm{bc}[u] / ((n - 1) \times (n - 2))$
    \EndFor
    \State \Return $\mathrm{bc}$
\EndFunction
\end{algorithmic}
\end{algorithm}

\begin{algorithm}
\caption{Memory Efficient Accumulation}
\begin{algorithmic}[1]
\Function{$\_$Accumulate\_Basic\_Mem\_Efficient}{$\mathrm{bc}, S, P, \sigma, s, \text{deg}_1$}
    \State $\delta_{\text{tmp}}[u], \zeta_{\text{tmp}}[u] \gets 0.0$ for all $u \in S$
    \While{$S \neq \emptyset$}
        \State $w \gets S.\text{pop}()$
        \State $\text{coeff} \gets \frac{1 + \delta_{\text{tmp}}[w]}{\sigma[w]}$
        \State $\text{coeff}_{\zeta} \gets \frac{\text{deg}_1[w] + \zeta_{\text{tmp}}[w]}{\sigma[w]}$
        \For{$v \in P[w]$}
            \If{$v \neq s$}
                \State $\delta_{\text{tmp}}[v] \gets \delta_{\text{tmp}}[v] + \sigma[v] \times \text{coeff}$
                \State $\zeta_{\text{tmp}}[v] \gets \zeta_{\text{tmp}}[v] + \sigma[v] \times \text{coeff}_{\zeta}$
            \EndIf
        \EndFor
        \If{$w \neq s$}
            \State $\mathrm{bc}[w] \gets \mathrm{bc}[w] + (1 + \text{deg}_1[s]) \times (\delta_{\text{tmp}}[w] + \zeta_{\text{tmp}}[w])$
        \EndIf
    \EndWhile
    \State \Return $\mathrm{bc}$
\EndFunction
\end{algorithmic}
\end{algorithm}

%% file: src/sampling.tex
Removing degree-1 nodes can improve both the run time and the accuracy of sampling BC estimators. Consider our proposed Algorithm~\ref{alg:oneround}, which performs a single round of peeling and samples $k$ nodes in the graph $\tilde{G}$ for SSSP computations, similar to methods in~\cite{wang2006fast,brandes2007centrality}. The pseudocode changes are annotated as comments. Specifically, in line 8, we iterate over all nodes in the sample of $k$ nodes. After the loop, the $\delta$ and $\zeta$ matrices must be rescaled by a factor of $\frac{\tilde{n}}{k}$.
Let's focus on the sum of $\delta$ and $\zeta$ terms in lines 18 and 20, excluding terms that are known or easy to compute, such as $n-2$ or $deg_1(u)$. Sampling introduces error in the $\delta$ and $\zeta$ terms computed on $\tilde{G}$. The unnormalized BC score $(n-1)(n-2) \mathrm{bc}(u)$ will involve a sum of the form 

\begin{align}
 \sum_{\substack{s \in V \\ s\neq u}} \delta_s^{(0)}(u) &= \sum_{\substack{s \in V_1 }} \delta_s^{(0)}(u) + \sum_{\substack{s \in V_{\geq 2} \\ s\neq u}}  \delta_s^{(0)}(u)     = 
 \sum_{\substack{s \in V_1 \\ y \in N(s) }} (\delta_y^{(1)}(u)+\zeta_y(u)) + \sum_{\substack{s \in V_{\geq 2} \\ s\neq u}}  (\delta_s^{(1)}(u)+\zeta_s(u)) \notag \\ 
   &= \sum_{\substack{s \in Y\\ s \neq u}} deg_1(s) (\delta_s^{(1)}(u)+\zeta_s(u)) + \sum_{\substack{s \in V_{\geq 2} \\ s\neq u}}  (\delta_s^{(1)}(u)+\zeta_s(u)) \label{eq:eq1} \\ 
 &=  \sum_{\substack{s \in V_{\geq 2} \\ s\neq u}} (1+deg_1(s)) (\delta_s^{(1)}(u)+\zeta_s(u))  \label{eq:eq2} 
\end{align}

Notice that we have expressed the original sum over $n-1$ source nodes $s$ as a weighted sum over $\tilde{n}-1$ source nodes in $V_{\geq 2}$. Before we delve into the analysis,  consider a graph where $V_1$ is a large fraction of the nodes, a realistic scenario according to Figure~\ref{fig:combined}. Sampling a small number  $k$ of pivots from $G$ will likely yield a sample that contains several degree-1 nodes. Such nodes yield overestimates of BC especially for their neighbors. Performing the sampling on the 2-core or even after 1-round of peeling which removes in practice the vast majority of degree-1 nodes, allows typically for a better estimation for the same sample size.  We provide a synthetic setting in Section~\ref{subsec:synth} that showcases this fact.

\noindent We formalize few scenarios that shows how peeling may help sampling. 

\begin{thm}
\label{thm:sample}
Let $G$ be a graph with $|Y|=O(\frac{\log \tilde{n}}{\epsilon^2})$. Then, sampling $\lceil \frac{\log \tilde{n}}{\epsilon^2} \rceil$ uniformly at random suffices to obtain an additive $\epsilon \frac{\tilde{n}-1}{n-1}$ approximation for all betweenness centralities with high probability. 
\end{thm}

\noindent In this scenario,  the input graph's $G$  set of nodes that have at least one degree-1 neighbor, i.e., the set $Y$ in line 6 of Algorithm~\ref{alg:oneround}, is small. Notice that compared to Brandes-Pich~\cite{brandes2007centrality}, we need $\Omega(\frac{\log \tilde{n}}{\epsilon^2})$ rather than $\Omega(\frac{\log n}{\epsilon^2})$ pivots to get the same guarantee.  Furthermore, instead of an $\epsilon$ additive approximation, we obtain an error equal to $\epsilon \frac{\tilde{n}}{n} \leq \epsilon$.  The same logic extends to multiple rounds of peeling.  Without any loss of generality, we assume that $\tilde{n}$ grows as a function of $n$, even arbitrarily slow (otherwise we can just compute the sum exactly in constant number of SSSP computations). In cases where $\tilde{n} \ll n$, the additive error tends to 0 as $n$ grows to infinity.

\begin{proof}
Consider equation \eqref{eq:eq1}. When $|Y|=O(\frac{\log \tilde{n}}{\epsilon^2})$ then we compute the first term 

$$\sum_{\substack{s \in Y\\ s \neq u}} deg_1(s) (\delta_s^{(1)}(u)+\zeta_s(u)),$$  

\noindent exactly by running $|Y|$ SSSP computations, using the vertices in $Y$ as the sources  We estimate the second term by sampling $k=\frac{\log \tilde{n}}{\epsilon^2}$ pivots uniformly at random. 

Notice that $\delta_s^{(1)}(u) \leq \tilde{n}-2$ and   $\zeta_s(u) \leq |V_1| = n-\tilde{n}$, so the sum of two such terms is upper bounded by $n-2$. We use Hoeffding's inequality by setting $M=n-2\approx n, \xi = \epsilon (n-2) \approx \epsilon n$,

$$ \Prob{ |\frac{\sum_{\text{pivot~} p} (\delta_p^{(1)}+\zeta_p)(u)  }{k}- \frac{\sum_{\substack{s \in V_{\geq 2} \\ s\neq u}}  (\delta_s^{(1)}+\zeta_s)(u) }{\tilde{n}-1} | \geq \epsilon(n-2) } \leq e^{-2 \log{\tilde{n}}} = \frac{1}{\tilde{n}^2}.$$

By taking a union bound over all nodes $u \in V_{\geq 2}$, we obtain that with high probability $1-\frac{1}{\tilde{n}}$ the above inequality holds for all nodes. Let's dissect what this inequality  suggests about the betweenness centrality of node $u$ in the input graph.  Our output estimator is $ \mathrm{est} = \frac{\tilde{n}-1}{k}  \cdot \sum_{\text{pivot~} p} (\delta_p^{(1)}+\zeta_p)(u)$ which is an unbiased estimator of the quantity $(n-1)(n-2)bc(u)$. This suggests that  whp

\begin{align*}
    (n-1)(n-2)bc(u)-\mathrm{est} &\in [ - \epsilon (\tilde{n}-1)(n-2), \mathrm{est}+ \epsilon (\tilde{n}-1)(n-2) ] \rightarrow \\
       bc(u)-\frac{\mathrm{est}}{(n-1)(n-2)} &\in [ - \epsilon \frac{\tilde{n}-1}{n-1}, \epsilon \frac{\tilde{n}-1}{n-1} ] 
\end{align*}

Assuming the conditions of our theorem are satisfied (note that determining the size of $|Y|$ can be done in linear time), the total runtime is $O(\frac{\log{\tilde{n}}}{\epsilon^2} (\tilde{n}+\tilde{m})+n \tilde{n}+ n \frac{\log{\tilde{n}}}{\epsilon^2})$.
\end{proof}

\noindent Notice that the runtime of our method can be significantly lower than the straightforward Brandes-Pich time complexity of \( O\left((n+m)\log{n}/\epsilon^2\right) \). Theorem~\ref{thm:sample} highlights multiple key benefits of peeling: it can reduce sample complexity, thereby decreasing overall runtime, and simultaneously enhance accuracy.

Let's consider now an alternate scenario, where $|Y|$ is large in size and thus Theorem~\ref{thm:sample} does not apply. As a special case, consider the extreme case where $Y=V_{\geq 2}$ and the degree-1 nodes are evenly spread among the $|Y|$ nodes, making the assumption that $deg_1(u) \approx \frac{|V_1|}{|Y|} = c$ $\forall u \in V_{\geq 2}$.  This simplified assumption allows us to apply Hoeffding's inequality in a meaningful way; we will explain later why  Hoeffding's inequality  does not yield good bounds when the degree-1 sequence is skewed and how to resolve this issue using Bernstein's inequality. Since the proof of the following statement follows closely the proof of Theorem~\ref{thm:sample}, we state it as a corollary.

\begin{cor}
\label{cor:allsame}
 Suppose all nodes in $Y$ have the same number of degree-1 neighbors. Then, sampling $\lceil \frac{\log \tilde{n}}{\epsilon^2} \rceil$ uniformly at random suffices to obtain an additive $\epsilon$ approximation for all betweenness centralities with high probability. 
\end{cor}

\begin{proof}[Proof Sketch]
Our goal is to estimate 

\begin{align}
\label{eq:su}
S_u  \stackrel{\text{def.}}{=} \sum_{\substack{s \in V_{\geq 2} \\ s\neq u}} (1+deg_1(s)) (\delta_s^{(1)}(u)+\zeta_s (u)) = (c+1) \sum_{\substack{s \in V_{\geq 2} \\ s\neq u}}  (\delta_s^{(1)}(u)+\zeta_s (u)).    
\end{align}

We apply Hoeffding's inequality and the rest of arguments as above for the sum $\sum_{\substack{s \in V_{\geq 2} \\ s\neq u}}  (\delta_s^{(1)}(u)+\zeta_s (u))$. Due to the term $(c+1)$ we obtain the following confidence interval for the betweenness centrality $bc(u)$ for all nodes $u \in V_{\geq 2}$ with high probability:

\begin{align*}
       bc(u)-\frac{\mathrm{est}}{(n-1)(n-2)} &\in [ - \epsilon \frac{(\tilde{n}-1)(c+1)}{n-1}, \epsilon \frac{(\tilde{n}-1)(c+1)}{n-1} ]. 
\end{align*}

This is where our simplified assumption is helpful. Since $c=\frac{|V_1|}{\tilde{n}}=\frac{n-\tilde{n}}{\tilde{n}}=\frac{n}{\tilde{n}}-1 \rightarrow (c+1) = \frac{n}{\tilde{n}}.$ This suggests that the confidence interval is defined by $\pm \epsilon \frac{(\tilde{n}-1)}{n-1}\frac{n}{\tilde{n}}\approx \epsilon$.

\end{proof}

\noindent \noindent When the degree-1 sequence is not uniform, for example skewed, with most of the nodes in $Y$ having low degree-1 and few having large degree-1,  Hoeffding's inequality is suboptimal. Specifically, let    $\Delta_1 \stackrel{\text{def.}}{=} \max_{v \in V} deg_1(v)$  to be the largest degree-1 value in the input graph $G$. Each term in the sum we want to estimate  is upper-bounded by $(1+\Delta)(n-2)$.  Observe  that we need to set $\xi = \epsilon \frac{ (n-1)(n-2)}{\tilde{n}-1} \approx \frac{n^2}{\tilde{n}} $ in Hoeffding's inequality to obtain an $\epsilon n (n-2)$ additive error term for the sum we care about or equivalently an asymptotically $\approx \epsilon$ error for the betweenness centrality of node $u$. This yields that the number of pivots $k$ has to be $\Omega\big(\frac{\log \tilde{n}}{\epsilon^2} \cdot \frac{\tilde{{n}} (1+\Delta_1) }{n} \big)$. An obvious bad scenario emerges when  $\Delta_1 = O(n)$ and $\tilde{n}=O(n)$, which can happen if all degree-1 nodes that are $O(n)$ are connected to a single node. A more refined expression can be obtained by Bernstein's inequality~\ref{prop:bernstein}.

To simplify the application of Bernstein's inequality, we will think of sampling each node $i \in V_{\geq 2}$ as a pivot with probability $\frac{k}{\tilde{n}-1} \approx \frac{k}{\tilde{n}}$ in order to obtain a sample size $k$ {\it in expectation} from $V_{\geq 2}$. Let's consider a node \( u \) that belongs to \( V_{\geq 2} \). For the purposes of this discussion, we will exclude it from the notation; as usual, once we adjust the failure probability for a single node,  we take a union bound over $\tilde{n}$ nodes to prove that the property holds for all with high  probability. In the following we make some reasonable asymptotic approximations  
to enhance readability and simplify certain expressions.

 Define $val_i
 \stackrel{\text{def.}}{=} (1+deg_1(i)) (\delta_i^{(1)}+\zeta_i)(u)$ and $S_u$ as in Equation~\eqref{eq:su} 
 and the set of random variables

 \[
\hat{X_i} = 
\begin{cases} 
 val_i & \text{with probability } p \stackrel{\text{def.}}{=} \frac{k}{\tilde{n-1}} \\
0 & \text{with probability } 1-p.
\end{cases}
\]

\noindent To apply Bernstein's inequality, we need to center the random variables around their mean. Define   

$$ X_i \stackrel{\text{def.}}{=} \hat{X}_i - p \cdot val_i.$$

Notice that $\Mean{X_i}=0, \Var{X_i} =  \Mean{X_i^2} = p\cdot (1-p) \cdot val_i^2$. For a small value of $k$, we approximate this as $\Mean{X_i^2}  \approx p \cdot val_i^2 = \frac{k}{\tilde{n}} \cdot val_i^2$. Furthermore, as argued earlier, $|X_i| \leq M=(1+\Delta_1)(n-2) \approx (1+\Delta_1) n $.  The event  

$$|\sum_{i} X_i| \leq t \Leftrightarrow  |\sum_{i} \hat{X}_i - \frac{k}{\tilde{n}} S_u | \leq t \Leftrightarrow  |\frac{\tilde{n}} {k} \sum_{i} \hat{X}_i -S_u| \leq t \frac{\tilde{n}} {k}  = |\text{est}_u-S_u| \leq t\frac{\tilde{n}} {k}, $$

\noindent where $\text{est}_u$ is our output estimator. We would like the error to be $\epsilon n(n-2)\approx \epsilon n^2$, where $0<\epsilon<1$ which implies that we need to set $t=\epsilon k  \frac{n^2}{\tilde{n}}$ (notice that without any loss of generality $\tilde{n}\geq k >0$, otherwise no sampling is needed). By plugging in these values in Bernstein's inequality we obtain  

\begin{align}
\label{eq:ugly}
\Prob{|\sum_{i} X_i| \leq \epsilon k  \frac{n^2}{\tilde{n}} } &= \Prob{|\text{est}_u-S_u| \geq \epsilon n^2} \leq  \exp \left( - \frac{1}{2} \frac{ \tfrac{\epsilon^2 k n^4}{\tilde{n}} }{  \sum_{i=1}^{\tilde n} val_i^2 + \frac{1}{3} (1+\Delta_1)  \epsilon n^3} \right) \\.
\end{align}

For example if $\tilde{n} = \sqrt{n}, \Delta_1 = \sqrt{n}$, we recover the guarantee of corollary~\ref{cor:allsame} that $k =\Omega( \frac{\log \tilde{n}}{\epsilon^2})$ pivots suffice to obtain a failure probability $\frac{1}{\text{poly}(\tilde{n})}$.  Overall, we obtain that 
$k = \Omega\left( \frac{\log \tilde{n}}{\epsilon^2} \cdot \tilde{n} \cdot\frac{\sum_{i=1}^{\tilde n} val_i^2 + \frac{1}{3} (1+\Delta_1)  \epsilon n^3}{n^4}\right)$ pivots are needed, an expression that is unfortunately a function of $\sum_{i=1}^{\tilde n} val_i^2$ that we do not know. However, knowing the degree-1  we can obtain non-trivial insights. Consider, for example, a power-law distribution for both the degree sequence (characterized by a slope $\alpha$) and the degree-1 sequence (characterized by a slope $r$). Recall that by definition

$$ \# \text{nodes with degree~} k = c \frac{n}{k^\alpha}.$$ 

\noindent Assuming $\alpha>2$, we obtain that $|V_1|=\Theta(n)$ and thus $\tilde{n}=\Theta(n)$. Furthermore, assuming   $r>3$ to ensure  bounded variance,  we obtain 

$$\sum_{i \in V_{\geq 2}} val_i^2= \sum_{i \in V_{\geq 2}} (1+deg_1(i))^2 n^2 \leq n^2 \sum_{i \in V_{\geq 2}} (2 \cdot deg_1(i))^2 \leq C(r) n^3,$$

\noindent for some constant $C$ that is a function of the degree-1 power law slope $r$. We upper bound $\Delta_1$ with the trivial upper bound $n$ (note that one can use a more refined bound of the maximum degree in $G$ that with high probability is at most $O(n^{1/\alpha})$). By plugging those values in Inequality~\eqref{eq:ugly}, the probability upper bound becomes 
$ \frac{\epsilon^2 \cdot k}{C+\frac{\epsilon}{3}}$
which gives a cleaner expression that implies that $\Omega(\frac{\log \tilde{n}}{\epsilon^2})$ pivots suffice. We present this special case in the following corollary, as it results in a constant factor improvement rather than an asymptotic one. However, improvements in constant factors can have a significant impact on real-world benchmarks.

\begin{cor}
\label{thm:powerlaw}
Let $G$ be a graph with power law degree with slope $\alpha>2$ and power law degree-1 sequence with slope $r>3$. A uniform sample of $\Omega(\frac{\log \tilde{n}}{\epsilon^2})$ suffices to obtain an $\epsilon$ additive approximation to all betweenness centrality scores with high probability.
\end{cor}

Understanding the sum \( S_u \) (Equation~\eqref{eq:su}) from an both an average case and extremal combinatorics perspective is an intriguing question, and finding more straightforward expressions for it is also of great interest that can lead to cleaner bounds for large classes of graphs.


%% file: src/exp.tex
\subsection{Experimental setup} 

\begin{table}[h!]
\centering
\begin{tabular}{|c|c|c|c|c|}
\hline
Dataset & $n$  & $m$ & $1- \frac{V_1^{(0)}}{n}$ & $\frac{\text{2-core size}}{n}$  \\
\hline
soc-dolphins  & 62 & 159 & 0.85 & 0.85\\
\hline
ca-CSphd  & 1\,882 & 1\,740  & 0.3 & 0.08\\
\hline
rt\_obama  & 3\,212 & 3\,423 & 0.2 & 0.11\\
\hline
soc-hamsterster  & 2\,426 & 16\,630 & 0.87 & 0.86\\
\hline
rt\_occupywallstnyc  & 3\,609 & 3\,833 & 0.11 & 0.08\\
\hline
soc-wiki-Vote  & 889 & 2914 & 0.78 & 0.76\\
\hline
ca-CSphd  & 1\,882 & 1\,740  & 0.3 & 0.08\\
\hline
inf-power  & 4\,941 & 6\,594 & 0.75 & 0.67\\
\hline
email-univ & 1\,133 & 5\,451 & 0.86 & 0.86\\
\hline
fb-pages-food  & 620 & 2\,102  & 0.8 & 0.77\\
\hline
ca-Erdos992   & 5\,094 & 7\,515 & 0.29 & 0.28\\
\hline
inf-euroroad  & 1174 & 1417 & 0.83 & 0.64\\
\hline
email-dnc-corecipient  & 906 & 10\,429 & 0.62 & 0.61\\
\hline
email-EuAll & 265\,214 & 365\,570 & 0.148  &  0.145\\
\hline
\end{tabular}
\caption{Number of nodes, edges, remaining fraction of nodes after one round of peeling, and the relative size of the 2-core for each dataset.}
\label{tab:data}
\end{table}

We use a variety of real-world graphs, publicly available at SNAP~\cite{snap} and Network Repository~\cite{nr}.  The characteristics of the  graphs we run the algorithms on are detailed in Table~\ref{tab:data}. As we can see from the last two columns, our observations from Figure~\ref{fig:combined} remain consistent: a significant portion of the nodes have a degree of 1, and most of these nodes are removed in the first round. To thoroughly understand the effect of removing degree-1 nodes on estimating betweenness centrality, our experimental section focuses on small-sized networks with up to 365,000 edges. Using these smaller networks allows us to obtain ground truth values for accurate comparison and analysis.  We run the experiments on a Mac with a 3.3 GHz Quad-Core Intel Core i5 processor and  8 GB 1867 MHz DDR3 of main memory. Our code is written in Python3 and is an adaptation of Networkx~\cite{hagberg2020networkx} betweenness centrality code. 


\subsection{Synthetic Experiment}
\label{subsec:synth}

We construct a graph $G$   with many degree-1 nodes and a core structure with nodes with high betweenness centrality. We generate a core of 50 nodes.  This core is divided into two subgraphs: one containing the first half of the nodes and the other containing the second half. A central node is chosen to connect the two subgraphs, serving as a node with distinctively high betweenness centrality. To create the connections, edges are added between each node in the first subgraph (excluding the central node) and the central node. Similarly, edges are added between the central node and each node in the second subgraph (excluding the central node). Within each subgraph, all nodes are fully connected to every other node, forming two complete subgraphs. Additionally, the graph includes 3000 nodes of degree 1. These nodes are connected to the core nodes in a skewed distribution, following a linearly decreasing pattern where nodes in the core are connected to a varying number of degree-1 nodes, excluding the central node which is not connected to any degree-1 nodes. This setup results in a total of 3000 degree-1 nodes connected to the core nodes. For visualization purposes, a subgraph comprising the first 500 nodes of \( G \) is extracted, including all 50 core nodes and a subset of the degree-1 nodes. The visualization in Figure~\ref{fig:synth}(a) uses a spring layout to emphasize the structure, with core nodes colored blue and degree-1 nodes colored red.    Figure~\ref{fig:synth}(b) plots the relative $\ell_1$ error of our method and Brandes-Pich. As we see, the 
Figure~\ref{fig:synth}(c) plots the BC estimates  using Brandes-Pich on the original graph~\cite{brandes2007centrality} and Algorithm~\ref{alg:oneround} with a sample of $k=10$ pivots. Both algorithms effectively identify the vertex with the highest BC score, and the estimates obtained are very close to the true values for all vertices.

\begin{figure}[t!]
    \centering
    \begin{tabular}{ccc}
        \includegraphics[width=0.25\textwidth]{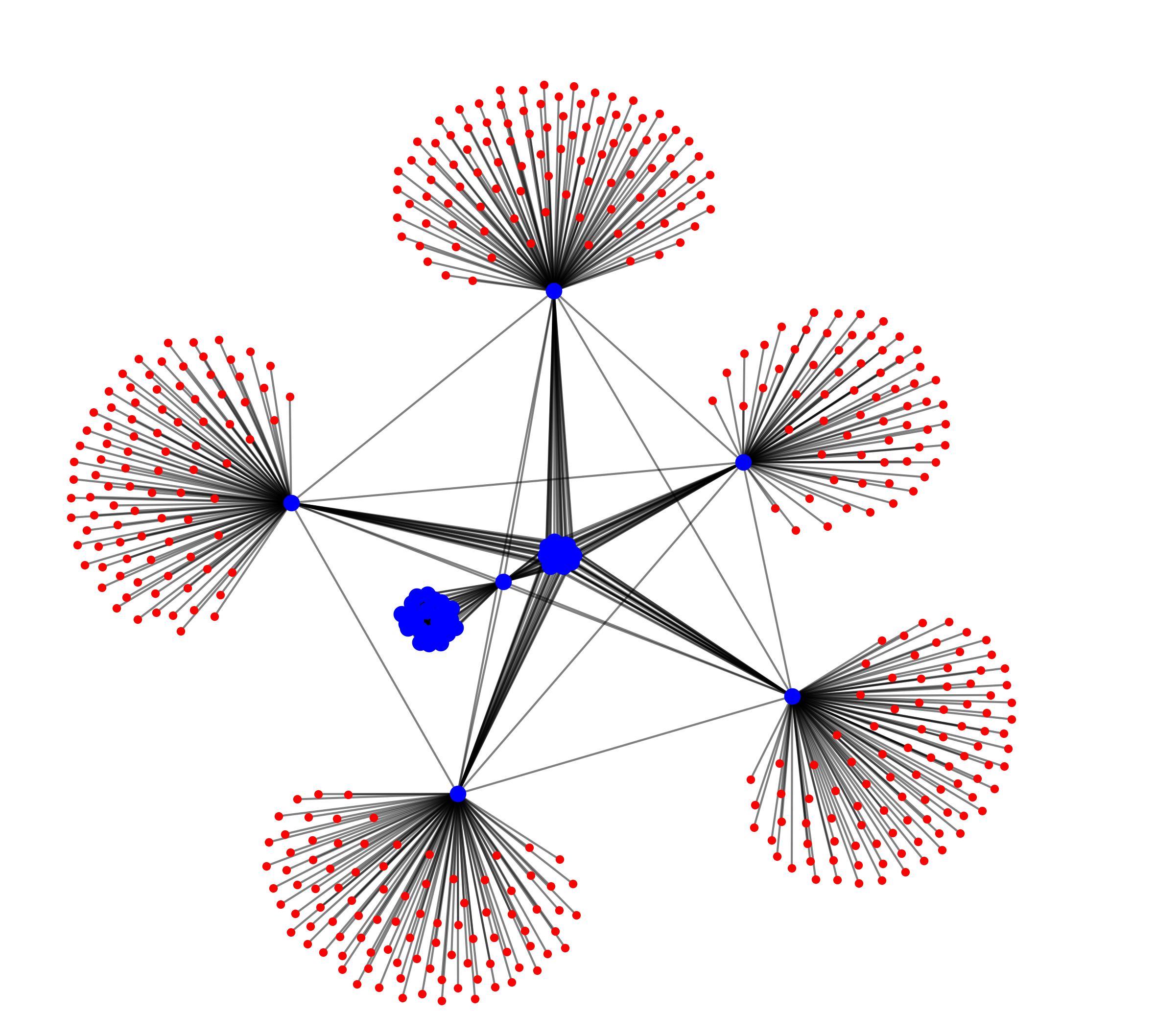} &
        \includegraphics[width=0.4\textwidth]{results/synthetic_dataset_range_k_relative_err.pdf}  &
        \includegraphics[width=0.35\textwidth]{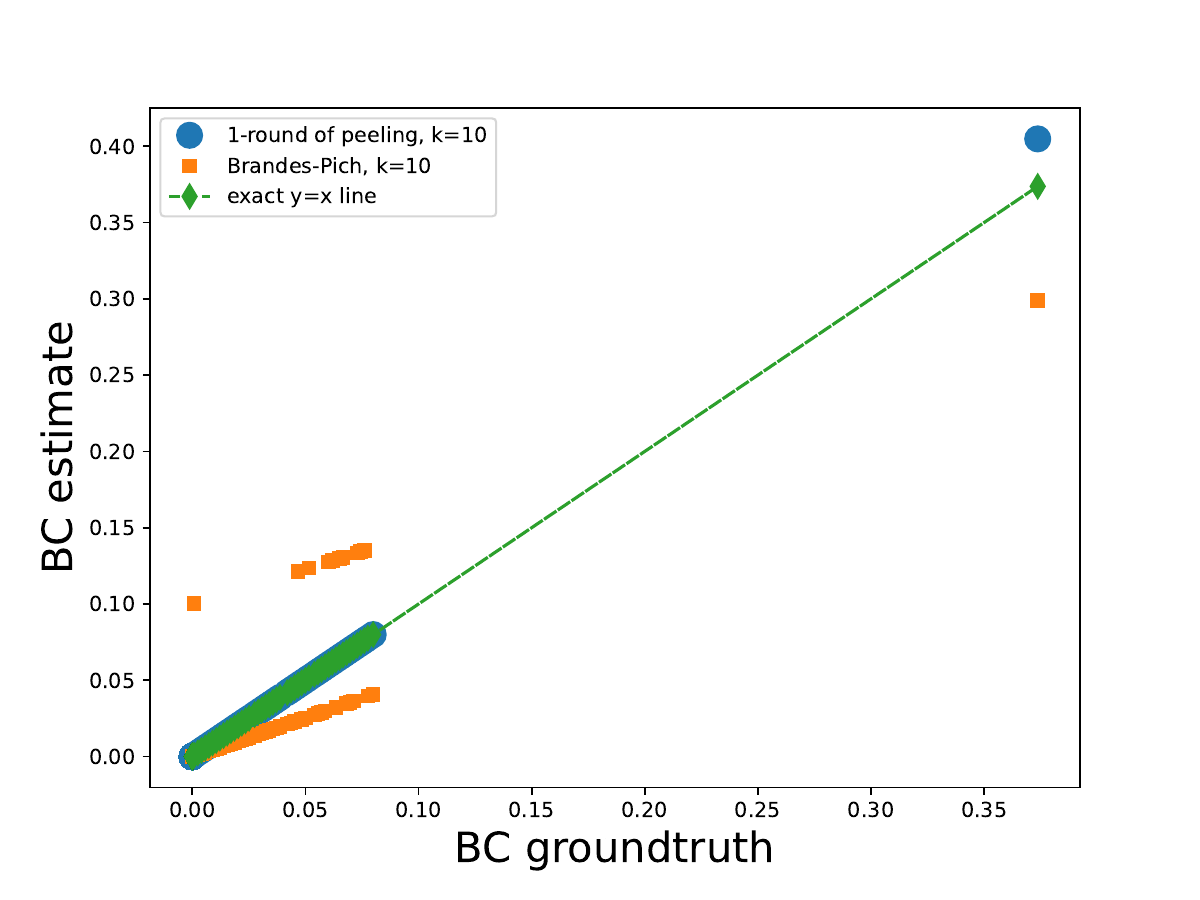}\\
                (a) & (b) & (c) \\ 

    \end{tabular}
    \caption{   (a)  Partial view of the synthetic graph with core and degree-1 nodes.  
          (b) Relative   $\ell_1$ error of Brandes-Pich and our method vs the number of sampled pivots $k \in \{10,20,\ldots,100\}$.  (c) BC estimates using a sample of 10 pivots with and without peeling vs groundtruth betweenness centralities.}
    \label{fig:synth}
\end{figure}


 \begin{figure}[t!]
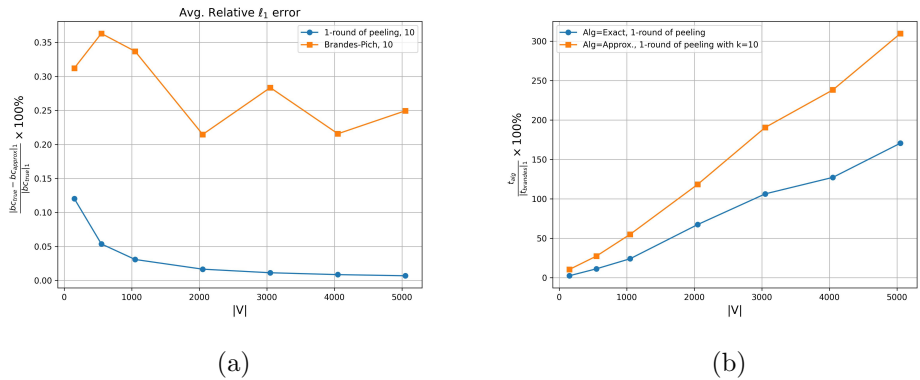

    \centering
    \begin{tabular}{cc}
        \includegraphics[width=0.4\textwidth]{results/synth_range_V1_relative_err} &
        \includegraphics[width=0.4\textwidth]{results/synth_range_V1_speedup.pdf}\\
                (a) & (b)   \\ 

    \end{tabular}
    \caption{ (a) For a fixed number of random pivots \( k = 10 \), the relative error of our method approaches 0, in contrast to the prediction by Brandes-Pich \cite{brandes2007centrality}, as stated in Theorem~\ref{thm:sample}.
(b) Speedups over the exact BC algorithm due to Brandes \cite{brandes2001faster} are achieved by both our exact memory-efficient implementation and our approximate version with a sample size of \( k = 10 \) pivots.}
    \label{fig:synth2}
\end{figure}

To illustrate Theorem~\ref{thm:sample} we slightly modify the above construction to make $Y$ a small set. We let the number of degree 1 nodes be a parameter $|V_1|$, while the core size is fixed to 50. Thus the total number of nodes is a variable equal to $n=50+|V_1|$. 
These degree-1 nodes are connected to the core nodes in a skewed distribution. Starting from the core nodes, excluding the central node (e.g., node 0), each core node $i$ starting from $i=1$ is connected to $|V_1|/2^{i+1}$ degree-1 nodes. The process continues until all degree-1 nodes are assigned or the remaining nodes are insufficient to follow the exact distribution pattern.  Figure~\ref{fig:synth2} presents several of our key findings. Specifically, Figure~\ref{fig:synth2}(b) demonstrates that the error of our approach, for a fixed number of samples, decreases as the number of nodes in \( G \) increases, with \(\tilde{n} = 50\) and only \(|V_1|\) increasing. These results, averaged over 5 runs, show that the Brandes-Pich algorithm exhibits a fluctuating error that does not decrease with \( n \). It is important to note that the BC scores of the degree-1 nodes are 0, so all the error arises from the core nodes. Given \( k = 10 \) and the increasing number of degree-1 nodes, the sample pivots eventually consist entirely of degree-1 nodes for a sufficiently large \(|V_1|\), resulting in significant relative errors for many nodes, except for those with the highest BC scores. Figure~\ref{fig:synth2}(b) shows how our 1-round of peeling, with and without sampling yields significant speedups over the standard BC computation.

\begin{figure}[t!]
    \centering
    \begin{tabular}{ccc}     
        \includegraphics[width=0.28\textwidth]{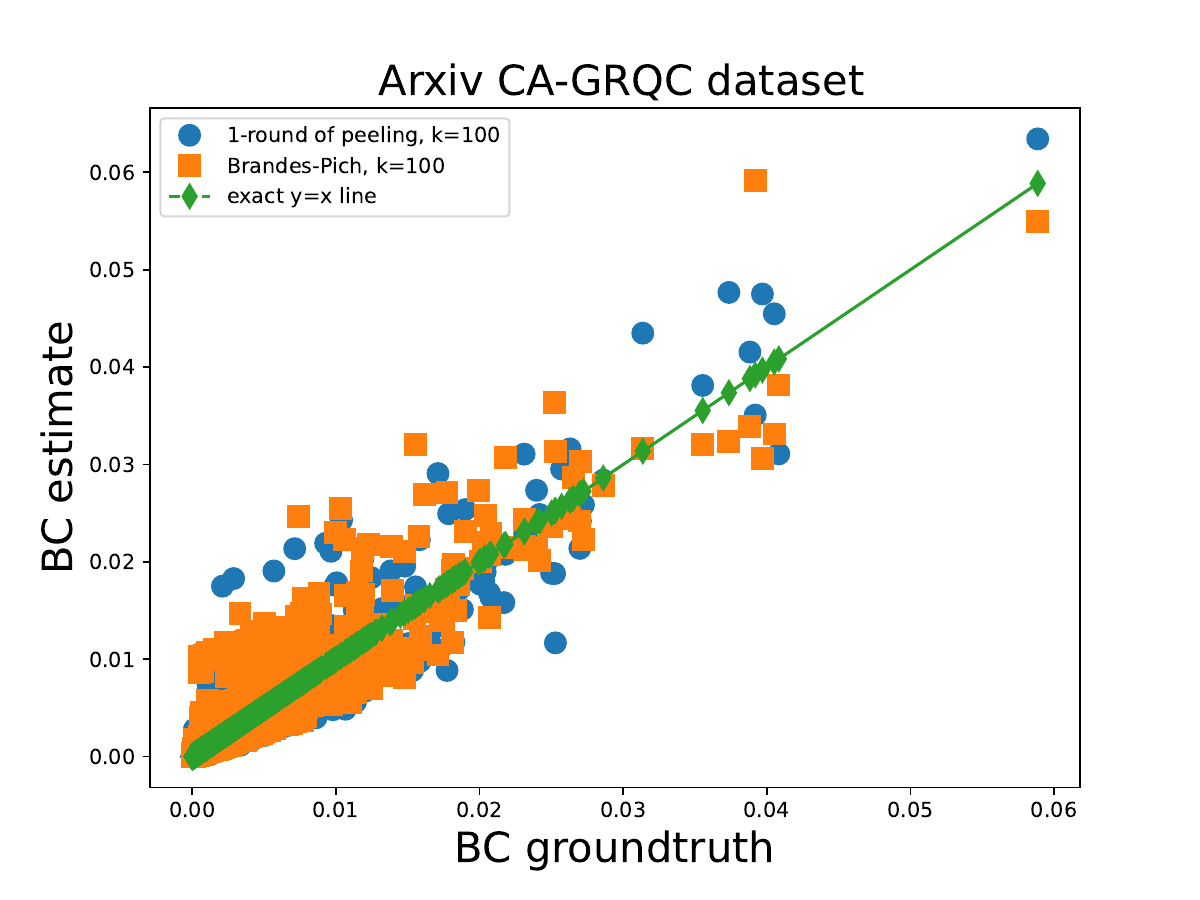} &
         \includegraphics[width=0.32\textwidth]{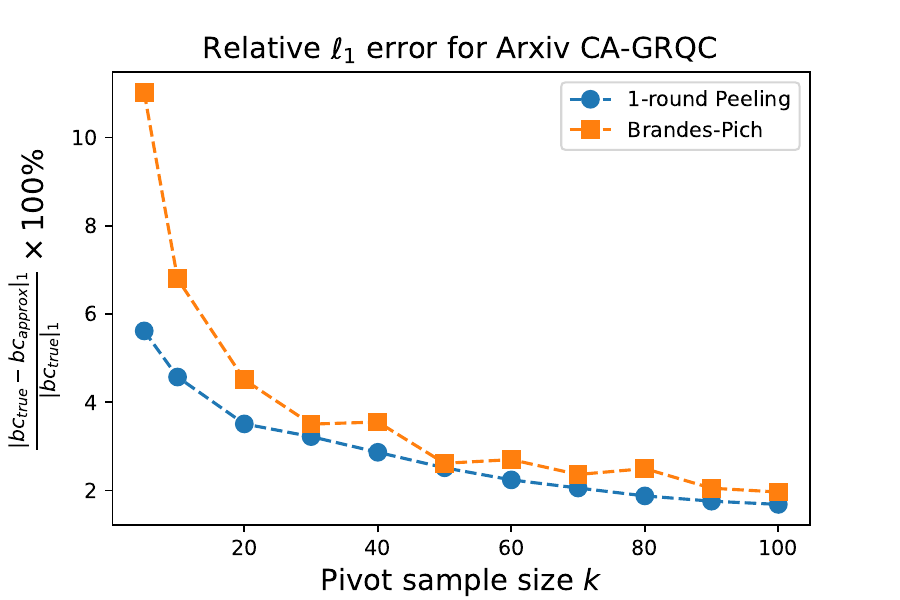} &
         \includegraphics[width=0.32\textwidth]{results/email-Eu-core_range_k_relative_err.pdf} 
        \\  
    \end{tabular}
    \caption{ \label{fig:res}   (a) BC estimates using a sample of size 100 with and without peeling vs groundtruth betweenness centralities for the Arxiv GR-QC dataset. (b) Average relative $\ell_1$ error of Brandes-Pich and our method over 5 runs vs the number of sampled pivots   for the Arxiv GR-QC, and (c)  the email-Eu-core datasets respectively.}
\end{figure}

\begin{figure}[!t]
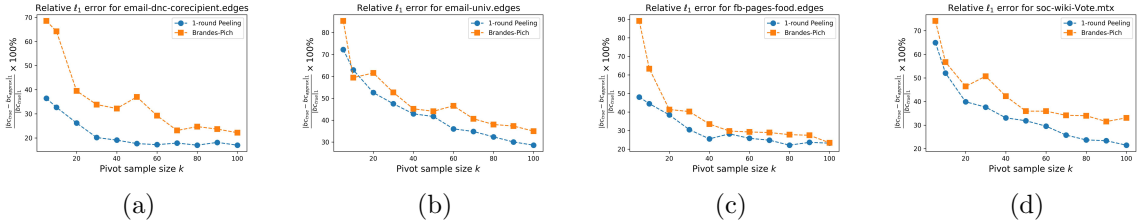

    \centering
    \begin{tabular}{cccc}
 
        \includegraphics[width=0.23\textwidth]{results/email-dnc-corecipient.edges_range_k_relative_err.pdf} & 
        \includegraphics[width=0.23\textwidth]{results/email-univ.edges_range_k_relative_err.pdf} 
        &
         \includegraphics[width=0.23\textwidth]{results/fb-pages-food.edges_range_k_relative_err.pdf} & 
         \includegraphics[width=0.23\textwidth]{results/soc-wiki-Vote.mtx_range_k_relative_err.pdf}
        \\
        (a) & (b) & (c) &  (d) \\
    \end{tabular}
    \caption{Average relative $\ell_1$ error of Brandes-Pich and our method vs the number of sampled pivots $k \in \{5,10,20,\ldots,100\}$ for four real-world networks from Table~\ref{tab:data} over five runs. }
    \label{fig:res2}
\end{figure}

\subsection{Real datasets} 
\label{exp:exact}

\paragraph{{\bf Exact BC computations}.}
We evaluate the performance of NetworkX, which utilizes a native Python implementation of Brandes' exact betweenness centrality algorithm, as referenced in Hagberg et al.~\cite{hagberg2020networkx}, against our memory-efficient algorithm (see Algorithm~\ref{alg:memefficient}). A principal observation, consistent with the conclusions of Sariyuce et al.~\cite{sariyuce2013shattering}, is that peeling accelerates the computation of betweenness centrality.  Specifically on average we obtain a {\bf 9.24}$\times$-fold speedup.

\paragraph{{\bf Approximate BC computations.}} We focus on assessing the accuracy of Algorithm~\ref{alg:memefficient} vs. Brandes-Pich sampling algorithm. 
In a nutshell, the empirical results we obtain are consistent: (i) Algorithm~\ref{alg:memefficient} generally exhibits lower error rates. (ii) This effect is more pronounced when the number of pivots is smaller. Figure~\ref{fig:res} presents some representative findings. In particular, Figure~\ref{fig:res}(a) displays a scatterplot comparing the betweenness centrality (BC) estimates from our algorithm and the Brandes-Pick method, both using $k=100$ for the collaboration network {\it Arxiv GR-QC}, against the actual BC scores. It is evident that both algorithms identify the node with the highest centrality, yet our Algorithm~\ref{alg:memefficient} consistently delivers more accurate estimates. This accuracy is highlighted in Figure~\ref{fig:res}(b), which illustrates the relative $\ell_1$ error of each algorithm for various numbers of pivots ($k=5, 10, 20, \ldots, 100$). At $k=100$, both algorithms show comparable accuracy, but the disparity grows with smaller $k$ values. These trends are echoed in the {\it email-Eu-core} dataset, where the results are even more marked due to the high number of degree-1 vertices in this network. Here, we explore a range of pivot counts from $k=5, 10, 20, \ldots, 100, 500, 1000$ until the relative error reduces to 25\% or lower. Figure~\ref{fig:res2} provides more plots of the average relative $\ell_1$ estimation error  as a function of the number of pivots $k$ for more real-world datasets; the plots are consistent with our findings for the {\it Arxiv GR-QC}  collaboration network and {\it email-Eu-core} communication network, showing that our method consistently achieves at least as low error as the Brandes-Pich algorithm for a given number of pivots, with the difference being pronounced for small number of pivots.


%% file: src/concl.tex
In this note, we explore the impact of removing degree-1 nodes from a graph to accelerate betweenness centrality computations~\cite{sariyuce2013shattering}. We present novel insights on an exact computation algorithm, demonstrating that this approach can significantly enhance the performance of approximate betweenness centrality computation in terms of both accuracy and speed. We provide both empirical and theoretical evidence. Notably, we demonstrate that removing degree-1 nodes leads from constant factor to asymptotic speedups depending on the size of the 2-core. Practically, leveraging the empirical observation that most degree-1 nodes in real-world networks are removed in the first round, we offer a memory-efficient betweenness centrality (BC) computation algorithm. This algorithm can be implemented by adding just a few lines to Brandes' exact code.

Our research raises a number of intriguing questions. While we explored uniform pivot sampling, utilizing pivots sampled with varying probabilities, along with other sampling methods as discussed in previous studies~\cite{riondato2018abra,riondato2014fast}, can prove highly effective in real-world applications. These approaches often result in reduced variance when employing Horvitz-Thompson and similar estimators~\cite{brewer2013sampling}. 
Additionally, expanding on the manipulations of bridges and articulation points as detailed by Sariyuce et al.~\cite{sariyuce2013shattering}, a compelling question arises regarding the definition of a new type of vertex separator. This type of separator would ensure that the shortest paths between node pairs on each side exclusively use nodes from their respective sides. Assuming such a separator can be found fast, this property could help speedup further BC computations. Also, obtaining cleaner bounds for other classes of graphs (e.g., bounded treewidth) remains an open question.

%% file: src/appendix.tex
\paragraph{{\bf SNAP datasets}} used in Figure~\ref{fig:combined}.

\begin{table}[h!]
\centering
\small
\begin{tabular}{cc}
\begin{tabular}{|c|r|r|}
\hline
SNAP Datasets~\cite{snap} & $n$ & $m$ \\
\hline
ca-HepPh & 12\,008 & 118\,521 \\
\hline
email-EuAll & 265\,214 & 365\,570 \\
\hline
wiki-Vote & 7\,115 & 100\,762 \\
\hline
ca-GrQc & 5\,242 & 14\,496 \\
\hline
ca-AstroPh & 18\,772 & 198\,110 \\
\hline
com-dblp.ungraph & 317\,080 & 1\,049\,866 \\
\hline
cit-HepTh & 27\,770 & 352\,324 \\
\hline
amazon0312 & 400\,727 & 2\,349\,869 \\
\hline
email-Enron & 36\,692 & 183\,831 \\
\hline
facebook\_combined & 4\,039 & 88\,234 \\
\hline
\end{tabular} &
\begin{tabular}{|c|r|r|}
\hline
ca-HepTh & 9\,877 & 25\,998 \\
\hline
web-Google & 875\,713 & 4\,322\,051 \\
\hline
soc-Slashdot0902 & 82\,168 & 582\,533 \\
\hline
amazon0302 & 262\,111 & 899\,792 \\
\hline
com-amazon.ungraph & 334\,863 & 925\,872 \\
\hline
cit-HepPh & 34\,546 & 420\,921 \\
\hline
soc-Slashdot0811 & 77\,360 & 546\,487 \\
\hline
ca-CondMat & 23\,133 & 93\,497 \\
\hline
\end{tabular}
\end{tabular}
\caption{Dataset description. All datasets are publicly available from SNAP~\cite{snap}. } 
\label{table:data}
\end{table}